\documentclass[conference]{ieeeconf}
\IEEEoverridecommandlockouts
\usepackage{cite}
\usepackage{amsmath,amssymb,amsfonts}
\usepackage{algorithmic}
\usepackage{graphicx}
\usepackage{textcomp}
\usepackage{xcolor}
\usepackage{gensymb}
\usepackage{svg}
\usepackage{url}
\usepackage{makecell}
\usepackage[percent]{overpic}
\usepackage{csquotes}
\usepackage{mathtools}

\def\BibTeX{{\rm B\kern-.05em{\sc i\kern-.025em b}\kern-.08em
    T\kern-.1667em\lower.7ex\hbox{E}\kern-.125emX}}

\newtheorem{prob}{Problem}
\newtheorem{thm}{Theorem}
\newtheorem{cor}{Corollary}
\newtheorem{prop}{Proposition}
\newtheorem{rem}{Remark}

\begin{document}

\title{Bounding Privacy Leakage in Smart Buildings}

\author{Rijad Alisic, Marco Molinari, Philip E. Par\'e, and Henrik Sandberg*\thanks{*Rijad Alisic, Philip E. Par\'e, and Henrik Sandberg are with the Division of Decision and Control Systems and Marco Molinari is with the Division of Applied Thermodynamics and Refrigeration at KTH Royal Institute of Technology, Sweden (e-mail: rijada, philipar, hsan, marcomo@kth.se)}}

\maketitle

\begin{abstract}
Smart building management systems rely on sensors to optimize the operation of buildings. If an unauthorized user gains access to these sensors, a privacy leak may occur. This paper considers such a potential leak of privacy in a smart residential building, and how it may be mitigated through corrupting the measurements with additive Gaussian noise. This corruption is done in order to hide the occupancy change in an apartment. A lower bound on the variance of any estimator that estimates the change time is derived. The bound is then used to analyze how different model parameters affect the variance. It is shown that the signal to noise ratio and the system dynamics are the main factors that affect the bound. These results are then verified on a simulator of the KTH Live-In Lab Testbed, showing good correspondence with theoretical results.% The variance of the estimated change time in the KTH Testbed is roughly an order of magnitude larger than the bound in many cases, making the bound tight.

%If the sensors are not secured properly, a potential adversary may gain access to the measurements

%, which causes a privacy leak.

%they may leak information private information about the occupants.

%These sensors handle private information about the occupants of the building. If the sensors are not secured properly, they may leak information private information about the occupants.
\end{abstract}

\section{Introduction} \label{sec:intro}
%Current trends of digitalization are pushing physical systems into this realm in order to improve user experience, reducing energy cost and reducing the carbon footprint. 

There is a growing trend of employing smart meters in the home, facilitating utility tracking in order to help lower costs. For example, choosing when to use energy to heat your home could be made more flexible with thermal batteries~\cite{scholtenpersistesi2017}. Automatic control for heating and cooling systems is a popular research topic. A common way to reduce energy usage has been to employ predictive controllers on HVAC systems~\cite{haokowlilinbarooahmeyn2013,parisioetal2014,barettlinder2015,raoukil2017}. Predictive controllers usually need user-generated in order to work properly.
%In order for predictive controllers to work here, user-generated data is needed.
This data contains information about the users' preferences, behaviors, and other private information. Historically, security measures against privacy breaches have mostly dealt with keeping an adversary from obtaining this data. Recent research focuses instead on how someone with access to the data may abuse it {for other purposes than what it was intended for}. A common concern is that there is no consensus around {the definition of} privacy, since it is application dependent.

In computer science, differential privacy is a concept that is widely used when protecting against those types of privacy attacks. A differentially private database can release structures and patterns of its data to anyone, but manages to keep the individual entries private. Typically, the released data is corrupted with additive noise from slowly decaying distributions~\cite{dwork2008}. The notion of differential privacy has been extended to dynamical systems as well. In~\cite{lenypapas2014}, the individual inputs to a multiple input system are kept from being estimated by injecting noise in the output.

Another example of privacy is considered in~\cite{farokhisandberg2018}. There a battery is placed between a household and a smart meter measuring energy consumption. The battery injects additive noise in order to corrupt the signal such that the estimation variance of the true consumption is made large{ through the minimization of the Fisher information}. These results are expanded for constrained noise in~\cite{farokhisandberg2019}, where the adversary is able to send queries to a more general database. %The notion of privacy here is to hide quantities that are directly measured, by introducing uncertainty in the variable that is requested.

In this paper, privacy leakage refers to the ability to infer user behavior from a set of measurements. These measurements are taken from a system that the user has interacted with. The user behavior is modeled as an input to a linear system, where the measurements are used in order to reconstruct this input. Input reconstruction is a well-researched problem, where conditions on what sensors are needed for reconstruction, so called input-output observability, have been investigated in, for example,~\cite{boukhobza2007,houpatton98,gracygarinkibangou18}.

%an adversary has to infer user behavior from a set of measurements. 

Here, on the other hand, the ability to reconstruct the input from measurements is not primarily limited by the sensors, but rather by how much the input is hidden by additive Gaussian noise on the measurements. We restrict ourselves to the cases of step inputs, turning the problem of input reconstruction into a multiple hypothesis test problem, known as the \emph{change point problem}. Change point problems are  well-researched; % since the second half of the 20th century; 
see~\cite{bassevillenikiforov1993} and~\cite{lehmanncasella1998}. %\textbf{In~\cite{willskyjones1976}, the \emph{Generalized Likelihood Ratio} method was developed in order to detect jumps in signals for fault detection.}
%Additionally, a convex relaxation of this method has been made and their result has since been expanded~\cite{ohlssongustafssonljungboyd2012}.
Defining privacy in the realm of hypothesis testing has been considered in~\cite{lioechtering2015,lioechtering2016,lioechtering2017}, with privacy being defined as missed detection.%the Type~II error.

A central result is that there exist a \emph{uniform minimum variance estimator} (UMVE) for detecting the size of the step changes, if the change time is known~\cite{deshayespicard1986}. The authors also show that a UMVE for unknown change times exists only if the sampling rate goes to infinity. However, for unknown change times and finite sampling density, no UMVE exists. The change time is typically estimated by using a UMVE on the amplitude change, conditioned on all possible change times. In this paper, privacy is defined as the variance of the estimated change time{, for which we provide a lower bound on in order to gain intuition on how to maximize privacy}. %It is possible to relate 
% this quantity 
%the variance of the estimated change time
%to the Type~II error for some types of estimators. However, this paper 
We focus on how the system dynamics affect the estimation variance of the input when the measurements are corrupted by additive Gaussian noise.

%If the change time is unknown and the time series has finite sampling density, then there is no

The paper is organized as follows, Section~\ref{sec:example} introduces the KTH Live-in Lab, which is a smart residential building. The attack model is defined and then an example of a privacy leak % how privacy of the inhabitants may be compromised 
is given. The section concludes by defining a defense strategy. Section~\ref{sec:probform} formulates the privacy problem before a solution is presented in Section~\ref{sec:theormodel}. Additionally, in Section~\ref{sec:theormodel}, a simple example depicts some important qualities that affect the occupancy change estimation. Section~\ref{sec:examplrevis} returns to the KTH Testbed in order to reconnect with the previous example and show how the results from Section~\ref{sec:theormodel} perform in a more realistic situation. Finally, the results are concluded in Section~\ref{sec:end}.% and future work is presented in Section~\ref{sec:future}.

\section{Motivating Example} \label{sec:example}

In this section we show an example of how privacy can be leaked through poorly protected sensors. First, let us introduce the KTH Testbed and how a privacy leak may occur. Subsequently we introduce a way to mitigate the leak.

%is introduced, which the simulator is based on. Then the attacker is defined, who has gained access to the measurements through a poorly protected sensor. The attacker's goal is to determine when the occupancy changes using the measurements. Thereafter, the defender is introduced.

\subsection{KTH Live-In Lab Testbed}
The data for this example is taken from the IDA ICE 4.8 software program for building simulations~\cite{idasimulator} of the KTH Live-In Lab Testbed, (Fig.~\ref{fig:poplll}),~\cite{liveinlabkth}. The KTH Testbed premises %feature a total of 305 square meters distributed over approximately 
120 square meters of living space, 150 square meters of technical space and a project office of approximately 20 square meters. In its current configuration the living space features four apartments, but the layout can be reconfigured for up to eight apartments.% to study optimal space use and interaction among the occupants. 
The KTH Testbed, which is part of the larger Live-In Lab testbed platform, is designed to be energetically independent, with dedicated electricity generation systems through photovoltaic panels, heat generation system (ground source heat pumps), and storage (electricity and heat). Sensors are extensively used to improve energy efficiency and indoor comfort, study user behavior, and to improve control and fault detection strategies.

% which is located in one of Einar Mattsson's three plus-energy buildings (Figure~\ref{}) on the KTH Campus

\begin{figure}
    \centering
    \includegraphics[scale=0.4,trim={0 0 0 0.7cm},clip]{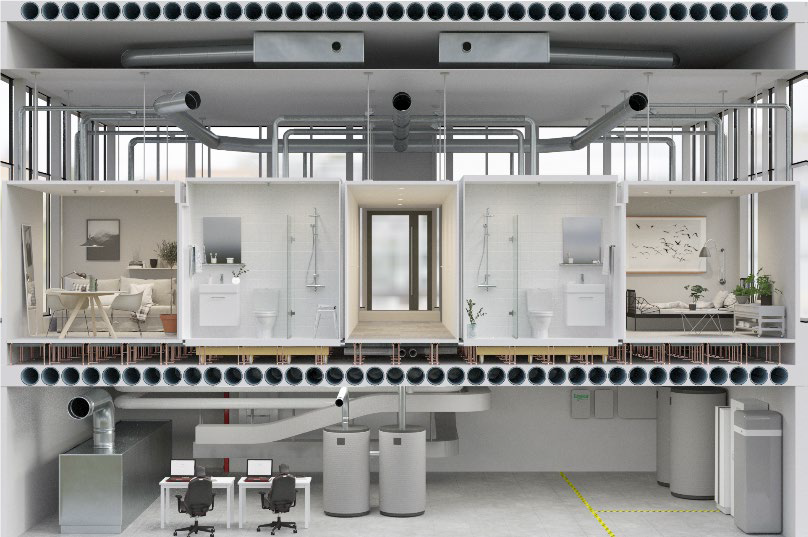}
    \caption{A cross-sectional view of the KTH Testbed, showing two separate apartments, the technical space below, and the HVAC system.}
    \label{fig:poplll}
\end{figure}
While the KTH Testbed has several types of sensors that could be used for detecting occupancy, we restrict ourselves to the temperature and relative humidity sensors here, because they are easy to correlate with the occupancy. Since there are multiple sensors of both types in each apartment, we use the mean of the sensors. Fig.~\ref{fig:apartfull} shows the temperature and relative humidity in a), and the occupancy in b) of Apartment~2. The data was obtained using the IDA ICE 4.8 simulator, sampled once every 9 minutes for a week. %While it is possible to see large increases in temperature and humidity with the naked eye, let us deploy the estimator to this data.

\begin{figure}
    \centering
    \begin{overpic}[width=1\linewidth]{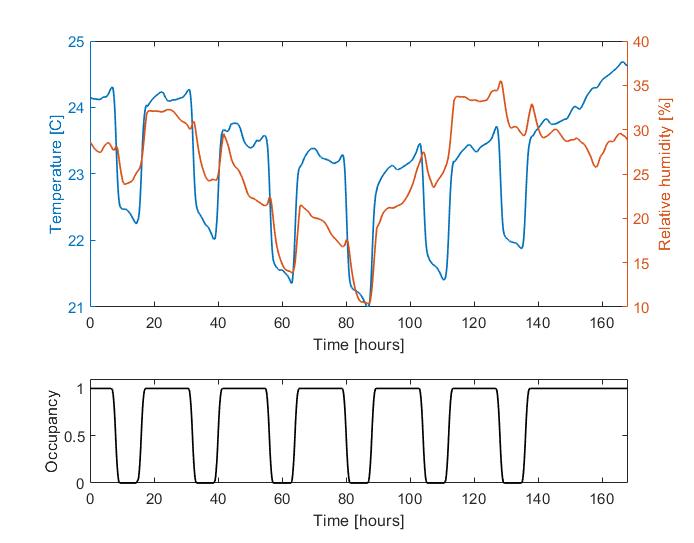}
   \put(2,73){a)}
   \put(2,24){b)}
    \end{overpic}
    \caption{The two graphs in plot a) show the temperature and relative humidity in Apartment~2 simulated for one week. The graph in plot b) shows the occupancy of the same apartment during the week.}
    \label{fig:apartfull}
\end{figure}

\subsection{Attack Model}\label{sec:attackmodel}
Let an adversary gain access to one of the types of sensors, either the temperature or humidity. Assume that the adversary has access to a linear time invariant model of the apartment dynamics,
\begin{equation}\label{eq:nonoisemodel}
M_0 : \begin{cases}
    x_{k+1} = Ax_k+Bu_k \\
    \hspace{11pt} y_k = Cx_k ,
\end{cases}  
\end{equation}
where $x_k \in \mathbb{R}^n$ is a vector of system states at each time instance $k$, $(u_k )_{k=0}^{N-1}$ is a sequence of scalar input signals representing occupancy, with $u_k \in \{0,1\}$ $\forall k$, and $y_k \in \mathbb{R}$ are the measurements produced at time instance $k$. The system matrices, $A \in \mathbb{R}^{n \times n}$, $B \in \mathbb{R}^{n \times 1}$ and $C \in \mathbb{R}^{1 \times n}$ define the dynamics of the sequence of states $(x_k)_ {k=0}^N$. Additionally, the measurement window length is denoted by $N$ and the $0$.

The system parameters of \eqref{eq:nonoisemodel} were determined from the input-output data of Apartment~1, 2 and 4 (see Fig.~\ref{fig:map}) using \mbox{MATLAB's} System Identification Toolbox. An 8-state system was obtained for both the temperature, denoted as $M^t$, and relative humidity, denoted as $M^h$. The IDA-ICE 4.8 simulator is non-linear, which makes \eqref{eq:nonoisemodel} an approximation of the true system. Apartment~3 was empty during the simulation, with the intention to capture the effect that the simulated weather has on the apartments, which the attacker also has access to. % and then remove these from the data of the other apartments. 
However, since the weather affects each apartment differently, it was not possible to completely remove its effects, creating some difficulties when identifying a good input-output model.

\begin{figure}
    \centering
    \includegraphics[width=1\linewidth,trim={7.6cm 4.7cm 7.6cm 3.2cm},clip]{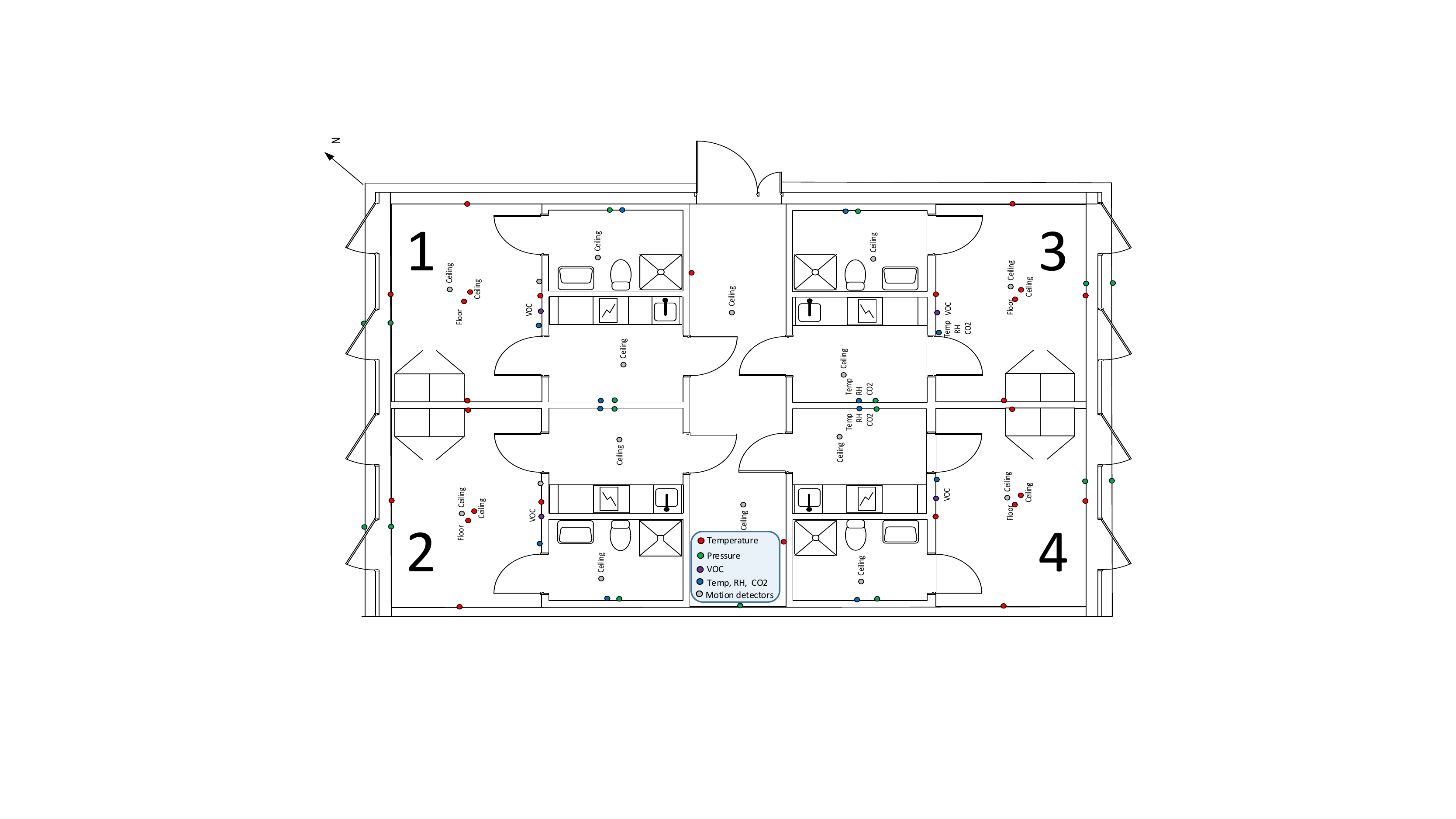}
    \caption{An overview of the apartments in the KTH Testbed.}
    \label{fig:map}
\end{figure}

\subsection{A Privacy Leak}

Consider a person entering their apartment, which has already reached a comfortable temperature and humidity range. Usually, entering a room is similar to adding a source of heat and humidity, almost instantly. One can model this change as a step input signal to the system \eqref{eq:nonoisemodel}, $u_k = 1 \in \mathbb{R}$ for $k \geq k^*$, and zero otherwise. The measurement of the step response at each time step, $y_k$, is sent to the controller.

Assume an adversary gains access to these measurements. In order to estimate when the occupancy changes, the adversary tries to solve a change point problem, as is described in Section~\ref{sec:intro}, since the signal $y_k$ behaves in one way up until $k^*$, and after that, it behaves differently. The adversary estimates the input by estimating the amplitude of the change through minimization of the $\ell_2$-norm, conditioned on all possible change times, $k^*$. Later, the measurements will be corrupted with additive Gaussian noise, and then the choice of this particular estimator is justified, since it becomes equivalent to the {Full Information Estimator~\cite{rawlings2013}}.% under constraints on the input signal.

With this {method}, the adversary estimates the occupancy change in Apartment~2. %when the occupant in Apartment~2 arrives home every day. 
The results of the estimation using the temperature sensor are shown in Fig.~\ref{fig:estnonoise}, together with the true occupancy of the apartment. One can see that the estimation correctly predicts the time step for when the person gets home.% Note that there is no additive noise here.

\begin{figure}
    \centering
    \includegraphics[width=1\linewidth,trim={4cm 11.3cm 4cm 11.3cm},clip]{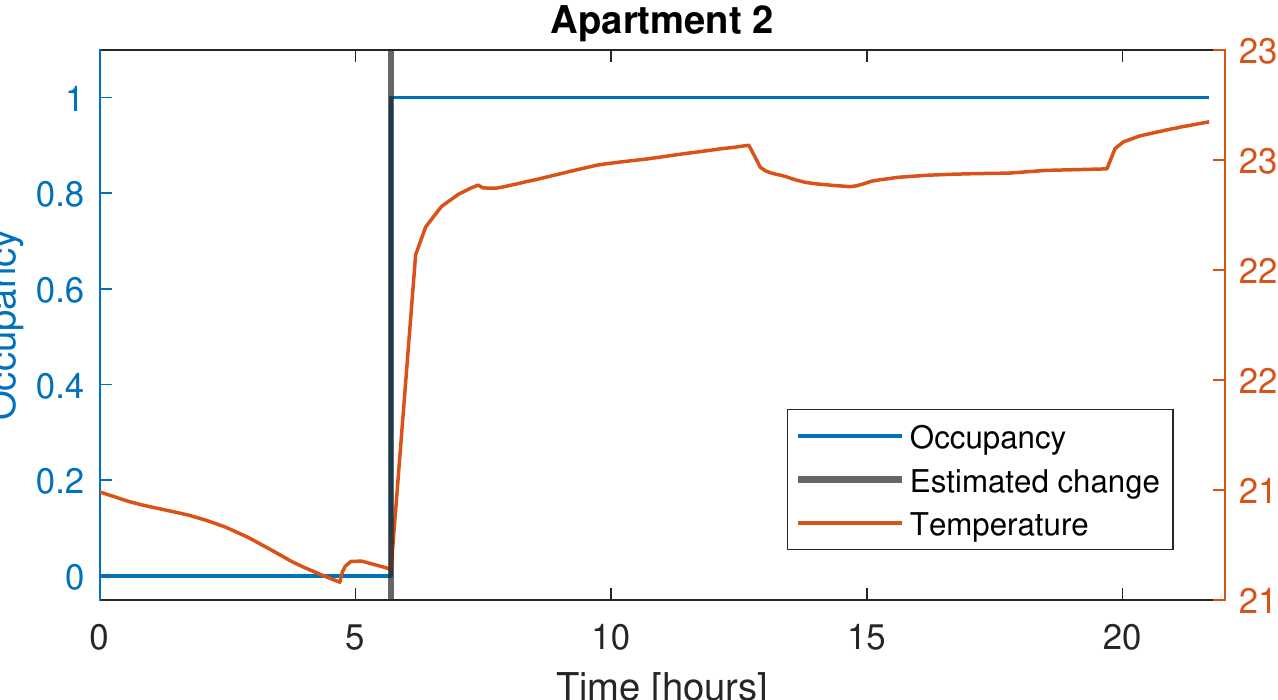}
    \caption{The graph shows the simulated change in occupancy and how it affects the temperature over 20 hours. The data is used to produce an estimate of the occupancy change.}
    \label{fig:estnonoise}
\end{figure}

\subsection{Defender Model}
Assume that the defender of the system has limited resources such that only one of the types of sensors in each apartment can be made secure against privacy leaks. The defender of the system needs to make a choice about which of the measurements to leave susceptible to privacy attacks. It is seen in Fig.~\ref{fig:estnonoise} that it is relatively easy to estimate the occupancy change unless the signal is somehow corrupted. In the next section, we propose that the measurements be corrupted with additive Gaussian noise in order to introduce uncertainty in the estimation of the occupancy change. This method is similar to earlier work; see~\cite{dwork2008} and~\cite{lenypapas2014}, which also propose corruptions using Gaussian or Laplacian noise in order to introduce privacy in the form of Differential Privacy. 

While adding a lot of noise to the sensor will hide the input, there is usually a utility function that the controller is trying to minimize, such as energy cost. Adding noise will generally increase this quantity. Therefore it is important for the defender to find a low noise level that ensures privacy.

\section{Problem Formulation} \label{sec:probform}
Consider again the model of the system the adversary desires to spy on \eqref{eq:nonoisemodel}. Let the measurements, defined as $Y = (y_k)_{k=0}^N$, $y_k\in \mathbb{R}$, be corrupted by additive Gaussian noise. The measurements are then generated by the system
\begin{equation} \label{eq:model}
M : \begin{cases}
    x_{k+1} = Ax_k+Bu_k \\
    \hspace{11pt} y_k = Cx_k  + e_k,
\end{cases}  
\end{equation}
where $ (e_k )_{k=0}^N$ is a sequence of zero mean, Gaussian noise with $\mathbb{E}\left [ e_ke_l\right] = \sigma^2 \delta_{kl}$ and $\delta_{kl}$ is the Kronecker delta.

%where $(u_k )_{k=0}^N$, is a sequence of scalar input signals representing occupancy, with $u_k \in \{0,1\}$ $\forall k$, $x_k \in \mathbb{R}^n$ is a vector of system states at each time instant, $k$ and $ (e_k )_{k=0}^N$ is a sequence of zero mean, Gaussian noise with $\mathbb{E}\left [ e_ke_l\right] = \sigma^2 \delta_{kl}$, where $\delta_{kl}$ is the Kronecker delta. The system matrices, $A \in \mathbb{R}^{n \times n}$, $B \in \mathbb{R}^{n \times 1}$, $C \in \mathbb{R}^{1 \times n}$ and $D \in \mathbb{R}$ define the dynamics of the sequence of states $(x_k)_ {k=0}^N$. Additionally, let the measurement window, $N$, be short enough so that there is only one step in the input, at time $k^*$.

%the input, $u_k$, only shifts once between its two possible values with the time step where the shift happens being denoted as $k^*$.

% rewrite this in terms of the definition of an estimator?

Now consider an estimator $\psi \left(Y,M \right)$, which may have access to the system model, $M$, that tries to estimate the sequence of inputs $(u_k)_{k=0}^{N-1}$ from the outputs $Y$. As stated in Section~\ref{sec:intro}, when the change time is not known, it needs to be estimated first before the amplitude of $u_k$ can be estimated. Guarding the change time from being estimated becomes the first line of defense, which also justifies why we only consider $u_k \in \{0 ,\, 1\}$. Additionally, since there are no UMVE for the change time, all possible estimators need to be considered. In this setting, we seek to find what factors affect the variance of the estimated change time.

\begin{prob} \label{prob:problem}
Consider an estimator of $k^*$, denoted by $\psi(Y,M)$, that has access to the model~\eqref{eq:model} and the measurements $Y$ of length $N$ such that $N \geq k^*$. What is the minimum variance that any such estimator can achieve?
%Assume that an estimator, $\psi$, has access to the model \eqref{eq:model}, the measurements $Y$ with a measurement window, $N$ such that a change in the input only happens once, at time $k^*$.{ What is the minimum variance of the estimate $\hat k$ of $k^*$ that any estimator can achieve?}
\end{prob}

\begin{rem}
%The answer to this problem will mostly depend on the system dynamics and the measurements. 
Implicitly, the model in \eqref{eq:model} gives the estimator only access to one type of sensor. Note that only looking at one sensor at a time is a restriction when performing the analysis. We plan to expand to include the combination of multiple sensors in future work.
\end{rem}

%%%%%%%%%%%%%%%% Add to future work
% While a specific sensor may leak a lot of information, a combination of two sensors that individually leak little information could, when combined, leak much more than the worse case sensor. Therefore, it is important to keep in mind that only looking at one sensor at a time is a restriction when performing the analysis and in future work, we will expand this to include combinations of multiple sensors.

\section{Main Results} \label{sec:theormodel}
Recall the definition of privacy as the variance of the estimated change time. Problem~\ref{prob:problem} asks how much privacy can be guaranteed in a privacy leak. In order to answer this question, we present the main theoretical result of this paper.

%\textbf{The Hammersley-Chapman-Robbins bound~\cite{hammersley1950} produces the best lower bound on the variance for an estimator of any parameters. Using this result, we can state the main theoretical result of this paper.}
% fix this for the D above
\begin{thm} \label{thm:main}
Consider any estimator of $k^*$, denoted by $\psi(Y,M)$, with bias $g(k^*)$. Then %The variance of the estimate of $k^*$, $\hat k$, using any estimator, $\psi(Y,M)$, with bias $g(\cdot)$, is lower bounded by,
\begin{equation} \label{eq:bound}
    \mathrm{Var}(\psi(Y,M)|k^*) \geq \max \limits_{\tau} \frac{(\tau + g(k^* +\tau) - g(k^*))^2}{\mathrm{e}^{{ \mathcal{S}(\tau,M) }}-1} \eqqcolon \mathcal{B}\left(M \right),
\end{equation}
for $\tau \in \{1, \,  \dots, \, N-k^*\}$, where
\begin{equation} \label{eq:sum}
    \mathcal{S}(\tau,M) = \frac{1}{\sigma^2}\sum \limits_{k= k^* + 1}^N\left(\sum \limits_{l= k^* }^{\mathrm{min}( k^* + {\tau} -1,k-1)} CA^{k-1-l}B \right)^2.
\end{equation}
\end{thm}

% \begin{thm} \label{thm:main}
% The variance of the estimate of $k^*$, $\hat k$, using any estimator, $\psi(Y,M)$, is lower bounded by,
% \begin{equation} \label{eq:bound}
%     \mathrm{Var}(\psi(Y,M)|k^*) \geq \max \limits_{\tau \neq 0} \frac{(\tau + g(k^* +\tau) - g(k^*))^2}{\mathrm{e}^{{ \mathcal{S} }}-1}
% \end{equation}
% for $\tau = \left\{ -k^*, \dots ,N-k^*-1 \right\}$. Additionally,
% \begin{equation} \label{eq:sum}
%     \mathcal{S} = \frac{1}{\sigma^2}\sum \limits_{k= k^* + \underline{\tau } +1}^N\left(\sum \limits_{l= k^* + \underline{\tau} }^{\mathrm{min}( k^* + \overline{\tau} -1,k-1)} CA^{k-1-l}G \right)^2
% \end{equation}
% where,
% \begin{equation*}
%     \overline{ \tau} = \mathrm{max}(0, \tau), \quad \underline{\tau} = \mathrm{min}(0,\tau).
% \end{equation*}
% \end{thm}

\begin{proof}
See the Appendix.
\end{proof}

%It is worth mentioning that the bound \eqref{eq:bound} is limited to single-input-single-output systems here in this paper, however, an extension to multiple-output-systems is straight forward. An extension to multiple input systems might require additional conditions, such as input-output observability as described by~\cite{}. 

It is worth reiterating that the bound in \eqref{eq:bound} is limited to single-input-single-output systems. Theorem~\ref{thm:main} illustrates all the important factors that should be considered when deploying a sensor that is susceptible to privacy attacks. In short, the bound is simply a function of the difference between the two most similar responses, in the $\ell _2$-norm. It gives insight into the important elements of linear systems that may cause information leaks. Let us analyze this bound and its implications through a simple, one-state linear time invariant system
\begin{equation} \label{eq:simpmodel}
M_1 : \begin{cases}
    x_{k+1} = ax_k+bu_k \\
    \hspace{11pt} y_k = cx_k  + e_k,
\end{cases}  
\end{equation}
where $x_k \in \mathbb{R}$, $y_k \in \mathbb{R}$, $u_k \in \{0, \, 1\}$, and $\mathbb{E} \left[ e_ke_l \right] = \sigma^2 \delta_{lk}$. The system parameters, $a$, $b$, and $c$ are scalar as well, with $a \geq 0$. For an unbiased estimator, the lower bound of the variance is given by the following corollary.

\begin{cor} \label{cor:cor} Consider any unbiased estimator of $k^*$, denoted by $\psi(Y,M_1)$. Then %The minimum variance any unbiased estimator of the change time can achieve for system \eqref{eq:simpmodel} is given by,
\begin{equation} \label{eq:simplebound}
    \mathrm{Var}(\psi(Y,M_1)|k^*) \geq  \frac{1}{\mathrm{e}^{{ \mathcal{S}(M_1)}}-1} \eqqcolon \mathcal{B}\left(M_1 \right),
\end{equation}
where
\begin{equation} \label{eq:simplesum}
    %\begin{aligned}
    \mathcal{S}(M_1) = \frac{1}{\sigma^2}\sum \limits_{k=k^*+1}^N \left( ca^{k-1-k^*}b \right)^2.%, \quad \text{if} \, a>0. %\\
    %\mathcal{S} & = \frac{1}{\sigma^2}\sum \limits_{k=k^*}^N \left( ca^{k-1-k^*}\left ( 1+{a} \right)b \right)^2, \quad \text{if} \, a<0.
    %\end{aligned}
\end{equation}
\end{cor}

\begin{proof}
The result follows directly from Theorem~\ref{thm:main} and
\begin{equation*}
    \left(\sum \limits_{l= k^* }^{\mathrm{min}( k^* + {\tau} -1,k-1)} ca^{k-1-l}b \right)^2 \geq\left( ca^{k-1-k^*}b \right)^2, \quad k > k^*.
\end{equation*}
\end{proof}

With these expressions, it is possible to analyse many properties of the one-state system and how the different parameters determine $\mathcal{B}(M_1)$.%how difficult it is to estimate the change time in the presence of noise.

\subsection{Signal to Noise Ratio} \label{sec:snr}

From the bounds, one can see that a large contributing factor is how large of an effect the occupancy has on the signal, as captured by the $\frac{c^2b^2}{\sigma^2}$ factor. {This quantity is related to the signal to noise ratio (SNR) squared}. Without loss of generality, $c$ can be set to 1, since it is equivalent to changing the noise variance by a corresponding factor. The ratio $\frac{b^2}{\sigma^2}$ gives a clear indication that there is a trade-off between the noise variance and how much the input affects the state. 

The $\frac{c^2b^2}{\sigma^2}$ factor is one design parameter the defender can determine, through the choice of the noise variance $\sigma^2$. Naturally, decreasing the SNR will increase the difficulty for the attacker to estimate the change time $k^*$. But when it comes to discriminating between the degree of privacy different sensors provide, the impact of SNR can be sidestepped, since different variances for different sensors can be chosen.

\subsection{System Dynamics and Sampling}
Another important component is the system dynamics, which essentially determine for how many time steps there exists a difference between different output signals. Intuitively, the longer the difference between two signals, the easier it is to differentiate between them in the presence of noise. Corollary \ref{cor:cor} sheds light on why this happens through the summation of the powers of $a$. 

\begin{prop}
Consider the single state system in \eqref{eq:simpmodel} with $a=1$. Then 
\begin{equation*}
    \mathrm{Var}(\psi(Y,M_1)|k^*) \geq \mathcal{B}(M_1) \to 0, \text{ when } N \to \infty. 
\end{equation*}
%\begin{equation*}
%    \mathrm{Var}(\psi(Y,M_1)|k^*) \geq \frac{1}{\mathrm{e}^{{ \mathcal{S}(\tau,M_1)}}-1} \to 0,
%\end{equation*}
\end{prop}

\begin{proof}
Insert $a=1$ in the summation \eqref{eq:simplesum}, then
\begin{equation*}
    \mathcal{S}(M_1)=\frac{c^2b^2}{\sigma^2}(N-k^*) \to \infty, \; \; \text{as } N \to \infty,
\end{equation*}
which in turn gives that the lower bound, $\mathcal{B}(M_1) \to 0$.
\end{proof}

%For example, a pure integrator system, $a=1$ gives that, $\mathcal{S}_N=\frac{c^2b^2}{\sigma^2}(N-k^*)$. The bound grows linearly as a function of the number of samples. 

Thus any good estimator will always be able to estimate the step change of integrator systems, given that enough samples have been collected. Because step changes that are initiated at different times never converge to the same value, the possibility to sample a large amount of useful data is created independently of the sample rate. Therefore, an estimator will be able to obtain zero variance. The bound also goes to zero for systems with unstable dynamics ($\Vert a \Vert >1$). In the unstable case, however, the rate of convergence to zero for the bound in \eqref{eq:simplebound} is greater the \enquote{more} unstable the system is.

In fact, the number of samples taken during the time steps when the sensor signal is different from other potential signals seems to be the fundamental quantity which determines the lower bound. The previous paragraph discusses this phenomenon for integrators and unstable systems. Asymptotically stable systems ($\Vert a \Vert < 1$) have a similar property.

\begin{prop} \label{prop:fasteig}
Consider two systems with only one-state as in \eqref{eq:simpmodel}, $M_1^1$ and $M_1^2$. Additionally, let $1 > a_2 > a_1 > 0$ with all other system parameters remaining the same. Then, the lower bounds satisfy
\begin{equation*}
    \mathcal{B}\left( M_1^1\right)>\mathcal{B}\left(M_1^2 \right).
\end{equation*}
\end{prop}
\vspace{2pt}
\begin{proof}
    The results follows from $\mathcal{S}(M_1^2) > \mathcal{S}(M_1^1)$.
\end{proof}

Proposition~\ref{prop:fasteig} shows that it is easier to estimate the change time in systems with slow dynamics, as opposed to systems with fast dynamics. Slower systems allow for more samples to be gathered during time steps when the signals are different. This phenomenon can be seen directly in the construction of the system matrix from zero-order hold sampling of a continuous linear time invariant system,
\begin{equation*}
\begin{aligned}
        \dot x(t) & = fx(t) + hu(t) \\ \Rightarrow x \left( \left(k+1\right)\Delta t\right) &= a x\left(k \Delta t\right) + b u\left(k \Delta t\right),
        \end{aligned}
\end{equation*}
where $a = \mathrm{e}^{-f \Delta t}$, and that $a \to 1$ (an integrator) as $\Delta t \to 0$. The number of samples over the same time window goes to infinity, causing $\mathcal{B}(M_1)$ to decrease. Thus the designer can influence the lower bound in \eqref{eq:simplebound} by changing the sampling rate.

Naturally, the number of samples can also be increased by increasing the time window, $N$, over which the sampling takes place. As stated previously, for integrator systems, $\mathcal{S}$ increases linearly and for unstable systems the summands increase geometrically. For stable systems though, the summands decrease geometrically so that, beyond a certain time step, the signal becomes similar to another signal that was initiated at a different time step. Samples beyond this time step will not produce any additional information to the estimator. Therefore, any additional sampling beyond this time step will no longer improve the estimation, which is of interest when designing a finite time horizon estimator.

\subsection{Sensor Design} \label{sec:sensorplacement}

Much of the intuition from the one-state system carries over to the multiple state systems with some additional technicalities. Specifically, the direction of vector $B$ and how it is transformed by $A$ affects the estimation. Typically, the designer of a system has some degree of choice when it comes to designing $C$, by placing the sensor in a particular place, which introduces another design parameter.

%Returning to the case of multiple states, as will be illustrated in the next section, much of the intuition from the one-state system carries over to multiple state systems with some additional technicalities, namely the direction of the vector $B$ after it has been filtered by the system. Typically, the designer of a system has some degree of choice when it comes to designing $C$, by placing the sensor in a particular place, or by choosing which sensor it should make private, which introduces another design parameter.

%For example, while it is true that the quality of the estimate could be improved by increasing the norm of vector $B$, there are additional things a designer needs to take into account, namely the direction of the vector after it has been filtered by the system. Typically, the designer of a system has some degree of choice when it comes to designing $C$, by placing the sensor in a particular place, or by choosing which sensor it should make private.

Denote the eigenvalues of $A$ by $\lambda_i$, $i \in \{1, 2, \dots, n\}$ with the corresponding eigenvectors, $v_i$ and where $|\lambda_{i}| \geq |\lambda_{i+1}|$. Assuming the eigenvectors $v_i$ form a basis, consider the components of $B$ with respect to this basis,
\begin{equation*}
    B=\sum \limits_{i=1}^{n}b_iv_i.
\end{equation*}

Using this basis, we can write the bound in \eqref{eq:bound} as,
\begin{equation} \label{eq:sumeig}
    \mathcal{S} = \frac{1}{\sigma^2}\sum \limits_{k= k^* + 1}^N\left(\sum \limits_{l= k^* }^{\mathrm{min}( k^* + {\tau} -1,k-1)} \sum \limits_{i=1}^{n}\lambda_i^{k-1-l} b_i Cv_i \right)^2.
\end{equation}

By designing $C$ in \eqref{eq:sumeig}, it is possible to change the bound on the variance in \eqref{eq:bound}. Choosing $C$ to be perpendicular to the largest eigenvalues such that $b_i\neq 0$, makes the sum converge as fast as possible. The privacy leak is then contained to the first few time steps. However, the variance is not necessarily increased, since the largest summands often occur in the beginning of the summation. This choice of $C$ might make these terms very large, thus making it easier to estimate the change.

%By choosing $C$ to be perpendicular to the eigenvectors of $A$ with the largest eigenvalues such that $B^\top v_i \neq 0$, makes the bound in \eqref{eq:simplebound} large if the time window, $N$, is large. The possible quality of estimation with many samples becomes much worse in this case. This may be seen by considering a summand containing a large power of $A$. From the power method, $v_1$ will be the dominating component in the $A^{M}B$ vector, for $M \gg 1$. Thereafter, $v_2$ will be the second most dominating eigenvector, and so on. By choosing $C$ to be perpendicular to as many of these components as possible, one will be able to minimize the possible quality of estimation for estimators that sample the data for a long period of time.

Conversely, one may design $C$ such that $CA^lB$ is small for the first few $l \geq 0$. The privacy leak in the first few time steps will then be minimal. However, the variance will not necessarily be decreased in this case either. If there is an eigenvalue that is very close to 1 which is not removed by $C$ and $b_i \neq 0$, then the sum \eqref{eq:sum} will be very large, due to the slow convergence of the summands to zero.

%Conversely, one may also look at the first summands in \eqref{eq:simplesum}, and design a $C$-vector that is perpendicular to as many of them as possible. This makes the variance of these estimates infinity for very small $N$, meaning that the estimator will not be able to detect the change in the signal for the first few samples. As $N$ grows large, the estimator will be able to produce the estimator that will be able to detect changes in the input with the best possible accuracy.

%% this needs to be rewritten when defender model is introduced! Talk about the trade-off in the previous paragraphs! %%% I don't see it as necessary anymore?
%With the tools above, a designer could use additional resources to prevent specific sensors such that it is very difficult to attack them. This could for example be by directly sending data to the regulator through a secured wired transmission, while letting the other sensors transmit their data wirelessly or through the Internet. While this gives the regulator data that such that it is able both produce a good estimate early on and in the long run, the designer chooses which property an adversary gets.

\section{Example Revisited} \label{sec:examplrevis}

Let us return to the example from Section~\ref{sec:example} of an adversary gaining access to one of the sensors in Apartment~2 of the KTH Testbed. 
The lower bound in \eqref{eq:bound} is calculated using the identified 8-state system from Section~\ref{sec:attackmodel}. In Table~\ref{tab:largetest}, the lower bound $\mathcal{B}$ for the two signals are presented with different $\sigma^2$. The empirical variance of the estimate, $\hat V$, of 1000 trials is shown as well. One may see that the bound holds in all cases, and the difference is less than an order of magnitude in some cases. Also, both $\mathcal{B}$ and $\hat V$ are large for small SNR, as was predicted in Section~\ref{sec:snr}. In Fig.~\ref{fig:temp}, we provide a realization of the measured step change in temperature from Apartment~2 using noise with $\sigma^2=2.25$, together with the simulated step change and the estimated step change for 1000 trials.

%the output at steady state, $y_{ss}=C(I-A)^{-1}B$.

%Using the input-output data displayed in Figure~\ref{fig:apartfull}, the adversary wishes to create a model of the system in order to create an estimator of the input. Using Matlab's System Identification Toolbox, an 8-state system for both the temperature and the humidity was chosen. 

%% FIX: table does not show SNR, fix empirical variance

\begin{table}
   \caption{SNR, Empirical variance $\hat V$ and the lower bound $\mathcal{B}$ under different variances on the additive noise}
    \label{tab:largetest}
    \begin{center}
\begin{tabular}{|c|c|c|c|c|c|}
\hline
    Model & \makecell{$\sigma^2$} &  \makecell{SNR} & \makecell{ $\hat V$ [min$^2$]} & \makecell{$\mathcal{B}$ [min$^2$]} \\ \hline
     Temperature, $M^t$ & 0.25  & 16.9 & 36.7 & 11.3 \\ 
     Humidity, $M^h$ & 0.25  & 204 & 11.9 & $3 \cdot 10^{-8}$ \\ 
     Temperature, $M^t$ & 2.25  & 1.87 & 3490 & 300 \\ 
     Humidity, $M^h$ & 2.25  & 22.6 & 48 & 7.74 \\ \hline
\end{tabular}
\end{center}
 
\end{table}

The empirical variance, shown in Table~\ref{tab:largetest}, is still larger than $\mathcal{B}$. One reason for this discrepancy is that there are more factors corrupting the data. For example, the data is generated together with a weather model. %, which also corrupts the data. 
%Here, it is assumed that the adversary has a model of how the weather affects the data, and simply subtracts this effect from the raw data they obtain from the attacked sensor. 
Privacy is somewhat increased due to the weather, since the attacker's weather model does not perfectly capture the impact that it has on Apartment~2. The effect of the weather can be modelled as process noise, which in turn, adds more stochasticity to the system, further masking the input signal. %This effect can also be explained through the SNR. In cases where the SNR is low, the main contributing factor for privacy is the measurement noise, making $\mathcal{B}(M)$ tight compared to cases where
Nonetheless, the discrepancy in the variance is less than an order of magnitude smaller for the temperature, which gives one an idea of how tight the bound is. Also, a small increase in the noise variance can increase the bound by a large magnitude, as seen in Table~\ref{tab:largetest}. Even though the SNR is very large in one case, the relative small addition of measurement noise increases the the empirical variance $\hat V$ by an entire time step. The large increase shows that model uncertainties may amplify the privacy introduced by measurement noise.

%However, in the presence of complex eigenvalues, this will not necessarily mean that it will be easier to estimate this change, since the induced oscillations might cause the sampling to be such that it looks like the signal is very similar to a signal with a step change at a different time.

\begin{figure}
    \centering
    \begin{overpic}[width=1\linewidth,trim={4cm 10.4cm 3.9cm 10.8cm},clip]{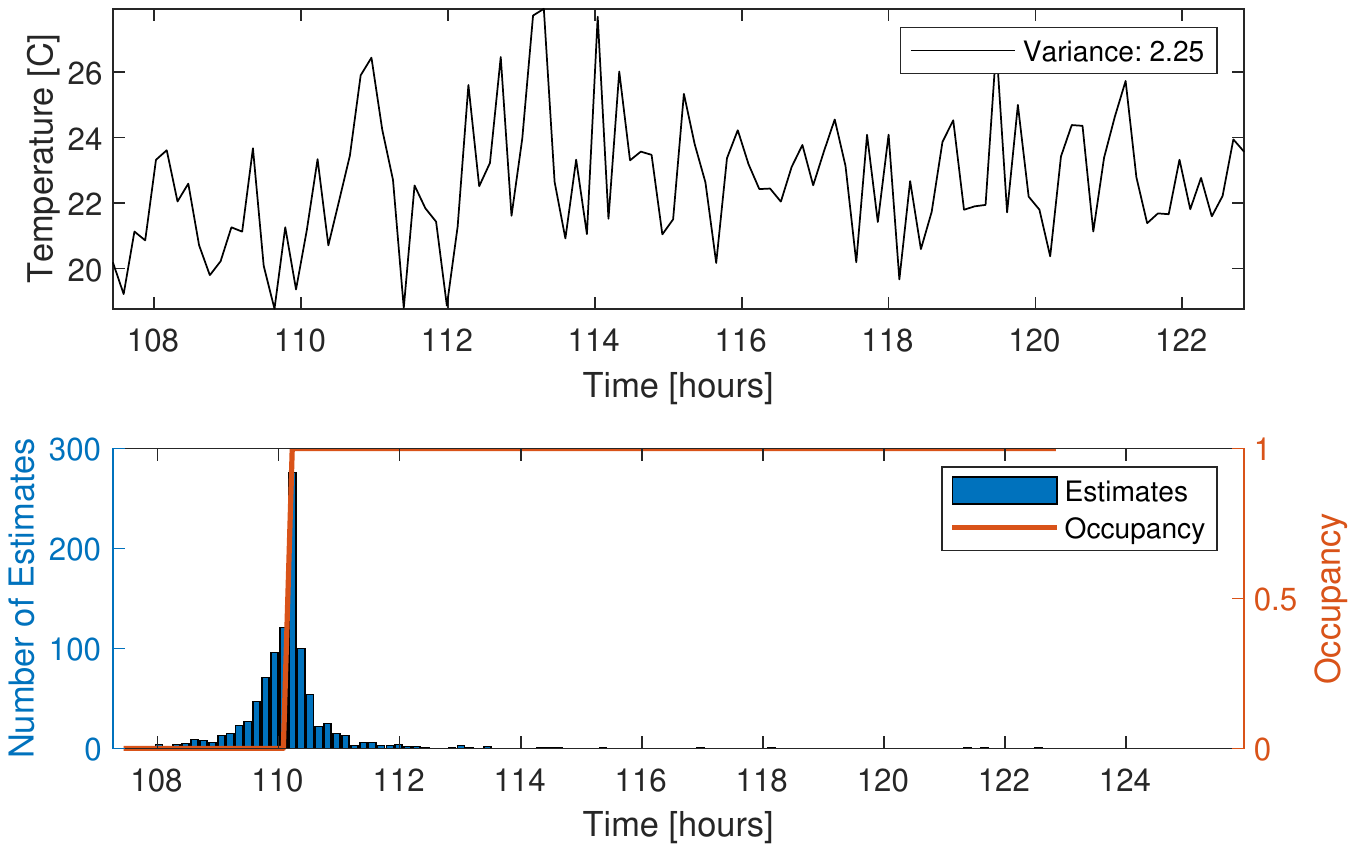}
    \put(10,60){a)}
    \put(10,25){b)}
    \end{overpic}
    \caption{In a), the temperature measurements are shown with additive Gaussian noise. In b), the histogram of 1000 estimates are shown together with the occupancy change.}
    \label{fig:temp}
\end{figure}

%\textbf{This paragraph and Figure~\ref{fig:rhum} does not add anything, I will remove it!}
%Consider the corresponding graph for the relative humidity in Figure~\ref{fig:rhum}. One may see that the estimators have much lower variance compared to the temperature, which is also reflected in the lower bound in Table~\ref{tab:largetest}. Upon closer inspection, one can see that the SNR is much larger for the humidity compared to the temperature, which is probably the main reason for the smaller variance of the change time in the humidity measurements. Therefore, a designer might think that the temperature signal is better to broadcast on the non-secured communication line. However, by increasing the noise in the humidity measurements until the SNR is equal for both humidity and temperature will the comparison be more fair.

%\begin{figure}
%    \centering
%    \includegraphics[width=1\linewidth,trim={4cm 11.1cm 4cm 11.1cm},clip]{A2Rhumsd15.pdf}
%    \caption{Caption}
%    \label{fig:rhum}
%\end{figure}

%Since both the empirical and lower bound of the variance of the time change are lower for the humidity, a designer would at a first glance see this signal as more vulnerable to privacy leaks. However, this gives a somewhat skewed impression that the temperature signal is better to broadcast on the non-secured communication line. Using the same additive noise on signals with different scales of magnitude will naturally destroy less information for the large magnitude signal. One may overcome this problem simply by scaling the variance of the noise by a corresponding factor.

One may see that the empirical variance is much lower for the humidity data compared to the temperature, which is also reflected in the lower bound in Table~\ref{tab:largetest}. Upon closer inspection, one can see that the SNR is much larger for the humidity compared to the temperature, which is probably the main reason for the smaller variance in the humidity measurements. A designer might think that the humidity sensor should be secured against privacy leaks. However the comparison is not fair, since the defender is not restricted to using the same noise for the different types of sensors.

%increasing the the noise in the measurements until the SNR is equal for both humidity and temperature will 

%the comparison be more fair.

%In order to overcome this issue and isolate the effect the system poles have on the estimation, we normalize all inputs and outputs so that they are between $0$ and $1$. Again, the estimation is done by \textbf{minimizing the data using the {mean-squared} error}. In Table~\ref{tab:smalltest}, the results of the estimation of the change time for a 1000 trials are performed. Both the lower bound $\mathcal{B}(M)$ and the empirical variance are again higher for the temperature compared to the humidity. These results imply that the defender should make the humidity sensors secure.

In order to overcome this issue and isolate the effect the system poles have on the estimation, we normalize all inputs and outputs so that they are between $0$ and $1$. Table~\ref{tab:smalltest} presents the results of estimating the change time for a 1000 performed trials. 
% are performed. 
Now, the lower bound $\mathcal{B}(M)$ and the empirical variance are lower for the temperature compared to the humidity. These results imply that the defender should make the temperature sensors secure.

By inspecting the slowest eigenvalues of the two models, one can see that the humidity model has a slightly larger maximum eigenvalue, and thus, should be easier to estimate. Nevertheless, the impact of the largest eigenvalues is not reflected in either $\mathcal{B}$ or $\hat V$. Recall the discussion in {Section~\ref{sec:sensorplacement}}; there it is argued that not only the largest eigenvalue determines the variance of the estimated change time, but also the corresponding eigenvectors, which should be a component of $B$ and not be perpendicular to $C^\top$. In fact, for both models model, $b_1 \approx 0$. The second largest eigenvalues in both models are also shown in Table~\ref{tab:smalltest}, whose eigenvectors are compatible with $B$ and $C$. Nonetheless, since the humidity model still is much slower than the temperature model, it should be easier to estimate. In fact, the eigenvalues impact all but the first summand of $\mathcal{S}(M)$, and a slow eigenvalue causes a large growth in $\mathcal{S}(M)$ as $N$ grows. In Table~\ref{tab:smalltest}, one can see that this growth for the humidity model is not large enough to even overcome the first summand in $\mathcal{S}(M^t)$. 

The privacy leak in $M^h$ for the first few time steps, is minimal since the corresponding $C^\top$ and $B$ are perpendicular to each other. For $M^t$, the privacy leak is concentrated on the first time step, which is so large that even many samples of the humidity would not be able to estimate a better change time. With these results, one can say that the temperature causes a larger privacy leak, when the defender applies the same SNR to both models.

\begin{table} 
   \caption{The two largest eigenvalues,  empirical variance $\hat V$ and the lower bound $\mathcal{B}$ when noise with SNR$=4$.}
    \label{tab:smalltest}
    \begin{center}
    \resizebox{\columnwidth}{!}{%
\begin{tabular}{|c|c|c|c|c|c|c|}
\hline
 Model & \makecell{$\big |\lambda_{1} \big|$} &  $\big |\lambda_{2} \big |$ & $\left(\frac{CB}{\sigma}\right)^2$ &$\mathcal{S}$ & \makecell{ $\hat V$ [min$^2$]} & \makecell{$\mathcal{B}$ [min$^2$]} \\ \hline
    %Type & \makecell{Variance \\ Noise} &  $|\lambda_{min}|$ & \makecell{Empirical \\ Variance} & \makecell{Lower bound \\ Variance}  \\ \hline
     $M^t$ & 0.995 & 0.84 & 0.134  &  0.142   & 342 & 101  \\ 
     $M^h$ & 0.996 & 0.97 & 0.018 & 0.088 & 545 & 184 \\ \hline
\end{tabular}
}
\end{center}
\end{table}

\section{Conclusions} \label{sec:end}
We have presented a method for analyzing how susceptible sensors in smart buildings are to privacy leakage in terms of occupancy change{, where by privacy we mean whether or not an adversary can infer user behavior from a set of measurements}. It is based on a fundamental lower bound on the variance of estimated change time of occupancy for any estimator. The lower bound was derived and then used to analyze the privacy properties of a one-state system. The results of this analysis showed how the SNR impacts the lower bound on the variance. Additionally, we analyzed the impact of the poles on the estimation, showing that unstable and integrator systems are the easiest to estimate, followed by slow stable systems; fast, stable systems are the hardest to estimate. Finally, it was shown that a slow eigenvalue is not enough to draw a conclusion about the privacy leakage. Thus, the defender can use these results to decrease the privacy leakage of a system or it can use the lower bound to ensure privacy from an unknown attacker. These results were then validated on data from a simulator of KTH Testbed using one particular type of estimator.

%\section{Future work} \label{sec:future}
As mentioned previously, an extension of the results to privacy leakage for combinations of sensors is a track for future work. Another extension is to investigate how process noise increases privacy. 
% Here, \textbf{only one estimator was employed} and the tightness of the bound in \eqref{eq:bound} was not discussed. The tightness will be explored in future research.
The tightness of the bound in \eqref{eq:bound}  will also be explored in future research.%\textbf{ Investigating how the covariance of identified model parameters affect the privacy leak will also be done in future work. Finally, research into the privacy-optimal control trade-off is of high importance for cyber-secure smart buildings and will be done in the future.}

\section*{Acknowledgements}
This work was supported by the Swedish Foundation for Strategic Research through the CLAS project (grant RIT17-0046). The authors would like to thank Alexandre Proutiere for 
his feedback.
%\textbf{[Do you want to add any of the professors' names here who you chatted with?]}

\bibliographystyle{unsrt}
\bibliography{bibliography.bib}

\newpage
\newpage
\onecolumn
%\section*{Appendix}

\appendix[Proof of Theorem~\ref{thm:main}]

The estimator, $\psi(Y,M)$, uses the measurements $Y$ and the model $M$ in order to estimate the change time $k^*$. Here, $k^*$ is a parameter which determines the probability distribution of $Y$. The minimum variance of the estimator is given by the Hammersley-Chapman-Robbins bound~\cite{hammersley1950},
\begin{equation*}
    \mathrm{Var}(\psi(Y,M) \big |k^*) \geq \sup \limits_{\tau \neq 0} \frac{\left( \mathbb{E}\left[ \psi(Y,M) \big| k^* + \tau \right] - \mathbb{E}\left[ \psi(Y,M) \big | k^* \right]  \right)^2}{\mathbb{E} \left[ \left( \frac{P\left(Y | k^*+\tau\right)}{P\left(Y | k^*\right)} -1 \right)^2 \bigg | k^*\right]},
\end{equation*}
where $\tau \in \mathbb{Z}$ and $P(Y|k^*)$ is the probability of obtaining measurements $Y$, conditioned on the change time is $k^*$. Evaluating the expectation in the denominator gives
\begin{equation} \label{eq:boundproof}
    \mathrm{Var}(\psi(Y,M) \big |k^*) \geq \sup \limits_{\tau \neq 0} \frac{\left( \mathbb{E}\left[ \psi(Y,M) \big| k^* + \tau \right] - \mathbb{E}\left[ \psi(Y,M) \big | k^* \right]  \right)^2}{\displaystyle\int \limits_{\mathbb{R}^N} \frac{P\left(Y | k^*+\tau\right)^2}{P\left(Y | k^*\right)} \mathrm{d}Y-1}.
\end{equation}

Since $\frac{P\left(Y | k^*+\tau\right)^2}{P\left(Y | k^*\right)} = \mathrm{e}^{2\log{P(Y_{}| k^*+\tau)}-\log{P(Y_{}| k^*)}}$, we write for $\tau \geq 1$,
\begin{align*}
    & \log{P(Y_{}| k^*+\tau)}-\log{P(Y_{}| k^*)} = \\ 
    & - \sum \limits_ {k=1}^N \frac{\left(y_k-CA^kx_0 - \sum \limits_{l = k^* + \tau}^{k-1}CA^{k-l-1}B \right)^2}{2 \sigma^2} + \sum \limits_ {k=1}^N \frac{\left(y_k-CA^kx_0 - \sum \limits_{l = k^*}^{k-1}CA^{k-l-1}B \right)^2}{2 \sigma^2} =  \\
    & -\frac{1}{2 \sigma^2} \sum \limits_{k=k^*+1}^N   (\sum \limits_{l = k^*}^{\mathrm{min}( k^* + \tau -1,k-1)} CA^{k-1-l}B)(2y_k-(2CA^kx_0 + \sum \limits_{l = k^*}^{k-1}CA^{k-l-1}B + \sum \limits_{l = k^* + \tau}^{k-1}CA^{k-l-1}B)).
\end{align*}
%where,
%\begin{equation*}
%    e_k = y_k-CA^kx_0-\sum \limits_{l=\hat k}^{k-1}CA^{k-l-1}G.
%\end{equation*}

Continuing, we see that,
\begin{align*} 
    & 2\log{P(Y_{}| k^*+1)}-\log{P(Y_{}| k^*)} = 2(\log{P(Y_{}| k^*+1)}-\log{P(Y_{}| k^*)})+\log{P(Y_{}| k^*)} = \\
    & -\frac{1}{\sigma^2} \sum \limits_{k=k^*+1}^N   (\sum \limits_{l=k^*}^{\mathrm{min}( k^* + \tau -1,k-1)} CA^{k-1-l}B)(2y_k-(2CA^kx_0  + \sum \limits_{l= k^*}^{k-1}CA^{k-l-1}B + \sum \limits_{l=k^* + \tau}^{k-1}CA^{k-l-1}B)) \\ 
    & - \sum \limits_ {k=1}^N \frac{\left(y_k-CA^kx_0 - \sum \limits_{l= k^*}^{k-1}CA^{k-l-1}B \right)^2}{2 \sigma^2} = \\
    & - \frac{1}{2\sigma^2}\sum \limits_{k=k^* +1}^N\left(y_k-CA^kx_0 - \sum \limits_{l=k^*}^{k-1}CA^{k-l-1}B + 2\left(\sum \limits_{l=k^*}^{\mathrm{min}( k^* + \tau -1,k-1)} CA^{k-1-l}B\right)\right)^2 \\
    & + \frac{1}{\sigma^2}\sum \limits_{k=k^* +1}^N\left(\sum \limits_{l=k^*}^{\mathrm{min}( k^* + \tau -1,k-1)} CA^{k-1-l}B \right)^2 - \sum \limits_ {k=1}^{k^*} \frac{\left(y_k-CA^kx_0 - \sum \limits_{l=k^*}^{k-1}CA^{k-l-1}B \right)^2}{2 \sigma^2}.
\end{align*}

Inserting this expression into the bound in \eqref{eq:boundproof}, evaluating the integral, and setting $\mathbb{E}\left[ \psi(Y,M) \big | k^* \right] = k^* + g(k^*)$, we obtain
\begin{equation}\label{eq:proofres}
    \mathrm{Var}(\psi(Y,M) \big |k^*) \geq \max \limits_{\tau \geq 1} \frac{ \left(\tau+g(k^*+\tau) - g(k^*)\right)^2    }{\mathrm{e}^{\mathcal{S}(\tau,M) }-1}, \text{ with } \quad \mathcal{S}(\tau,M) = \frac{1}{\sigma^2}\sum \limits_{k=k^* +1}^N\left(\sum \limits_{l=k^*}^{\mathrm{min}(k^* + \tau -1,k-1)} CA^{k-1-l}B \right)^2.
\end{equation}

For $\tau \leq -1$, an equivalent calculation can be made giving the same expression as \eqref{eq:proofres}, but replacing $\mathcal{S}$ with
\begin{equation*}
    \mathcal{S}_-(\tau,M) = \frac{1}{\sigma^2}\sum \limits_{k=\hat k+\tau +1}^N\left(\sum \limits_{l=\hat k+\tau}^{\mathrm{min}(\hat k -1,k-1)} CA^{k-1-l}B \right)^2.
\end{equation*}

However, since
\begin{equation*}
    \mathcal{S}(|\tau|,M) \leq \mathcal{S}_-(-|\tau|,M),
\end{equation*}
for each positive integer $|\tau|$, we can ignore the $\tau \leq -1$ cases. \hspace*{\fill}\QED

\end{document}